\newif\ifpublic
\publictrue
\ifpublic
\documentclass[11pt,a4]{article}
\else
\documentclass[11pt,draft, letterpaper]{article}
\fi
\PassOptionsToPackage{hyphens}{url}

\usepackage{breakurl}
\usepackage[margin=1in]{geometry}
\usepackage[final,pagebackref,colorlinks=true,linkcolor=blue,urlcolor=blue,citecolor=blue,pdfstartview=FitH,breaklinks=true]{hyperref}
\usepackage[T1]{fontenc}    
\usepackage{amsmath}       
\usepackage{amssymb}
\usepackage{amsthm}
\usepackage{amsbsy}
\usepackage{nicefrac}       
\usepackage{lipsum}
\usepackage{fullpage}
\usepackage{mathpazo}
\usepackage{mathtools}
\usepackage{enumitem}
\usepackage{subfiles}
\usepackage{xcolor}
\usepackage{setspace}
\usepackage{algpseudocode}
\usepackage[linesnumbered,ruled, lined, commentsnumbered, longend]{algorithm2e}
\ifpublic
\usepackage[disable]{todonotes}
\else
\usepackage[colorinlistoftodos]{todonotes}
\fi
\newcommand{\phnote}[1]{\todo[color=red!100!green!33, size=\footnotesize]{ph: #1}}
\newcommand{\sidnote}[1]{\todo[color=green, size=\footnotesize]{sid: #1}}
\newcommand{\sayantannote}[1]{\todo[color=green,size=\footnotesize]{sc: #1}}

\newcommand*\samethanks[1][\value{footnote}]{\footnotemark[#1]}

\usepackage[capitalize,nameinlink]{cleveref}
\crefname{algocf}{alg.}{algs.}
\Crefname{algocf}{Algorithm}{Algorithms}
\renewcommand{\eqref}[1]{\hyperref[#1]{(\ref*{#1})}}
\usepackage{stmaryrd}
\usepackage{MnSymbol} 
\usepackage{tikz}
\usepackage{pgfplots}
\usepackage{subcaption}
\usepackage{lipsum}
\usepackage{etoolbox}
\usepackage{appendix}
\AtBeginEnvironment{figure}{\addvspace{7mm}}\AtEndEnvironment{figure}{\addvspace{7mm}}
\pgfplotsset{compat=1.14}
\theoremstyle{plain}
\newtheorem{theorem}{Theorem}[section]

\newtheorem{lemma}[theorem]{Lemma}

\newtheorem{claim}[theorem]{Claim}

\theoremstyle{definition}

\newcommand\blfootnote[1]{%
  \begingroup
  \renewcommand\thefootnote{}\footnote{#1}%
  \addtocounter{footnote}{-1}%
  \endgroup
}

\renewcommand{\epsilon}{\varepsilon}

\renewcommand{\phi}{\varphi}

\newcommand{\poly}{\mathrm{poly}}

\newcommand{\set}[1]{\{#1\}}
\newcommand{\brak}[1]{\lbrace#1\rbrace}
\DeclareMathOperator{\E}{\mathbb{E}}

\newcommand{\warmup}[1]{\textsc{Compress}#1}
\newcommand{\warmuphelper}[1]{\textsc{SpruceUp}#1}

\newcommand{\sampler}[1]{\textsc{contract}#1}

\newcommand{\gd}{\textsc{Glauber Dynamics}}

\newcommand{\out}{\mathsf{f}}
\DeclarePairedDelimiter{\abs}{\lvert}{\rvert}
\title{Improved Bounds for Perfect Sampling of $k$-Colorings in Graphs}
\author{
  Siddharth Bhandari\thanks{Tata Institute of Fundamental Research,
    Mumbai. email: {\tt
      \{siddharth.bhandari,sayantan.chakraborty\}@tifr.res.in}. Research of
    the authors supported by the Department of Atomic Energy,
    Government of India, under project no. 12-R\&D-TFR-5.01-0500. The research of the first author was supported in part by Google PhD Fellowship. }
   \and
 Sayantan Chakraborty\samethanks
}

\begin{document}
\maketitle

\begin{abstract}
 We present a randomized algorithm that takes as input an undirected $n$-vertex graph $G$ with maximum degree $\Delta$ and an integer $k > 3\Delta$, and returns a random proper $k$-coloring of $G$. The 
 distribution of the coloring is \emph{perfectly} uniform over the set of all proper $k$-colorings; the expected running time of the algorithm is $\mathrm{poly}(k,n)=\widetilde{O}(n\Delta^2\cdot \log(k))$.
 This improves upon a result of Huber~(STOC 1998) who obtained a polynomial time perfect sampling algorithm for $k>\Delta^2+2\Delta$.
 Prior to our work, no algorithm with expected running time $\mathrm{poly}(k,n)$ was known to guarantee perfectly sampling with sub-quadratic number of colors in general.
 
 Our algorithm (like several other perfect sampling algorithms including Huber's) is based on  the Coupling from the Past method. Inspired by the \emph{bounding chain} approach, pioneered independently by 
 Huber~(STOC 1998) and
 H\"aggstr\"om \& Nelander~(Scand.{} J.{} Statist., 1999), we employ a novel bounding chain to derive our result for the graph coloring problem. 
 \end{abstract}
\blfootnote{A preliminary version of this paper (\href{https://arxiv.org/abs/1909.10323v1}{arXiv:1909.10323v1}) proved a weaker result that achieves expected polynomial time when $k>2e\Delta^2/\ln(\Delta)$.}

\section{Introduction}\label{sec:Introduction}

A $k$-coloring of a graph is an assignment of colors from the set
$[k]=\{1,2,\ldots,k\}$ to the vertices so that adjacent vertices are assigned different colors. We consider the problem of randomly sampling colorings of a given graph. The input is a graph $G$ 
and an integer $k$: our goal is to generate a $k$-coloring uniformly at random from the set of all $k$-colorings of $G$.
The problem of sampling $k$-colorings has several implications in theoretical computer science and statistical mechanics. For example, Jerrum, Valiant and Vazirani~\cite{JerrumValiantVazirani} show that from an almost uniform sampler for proper $k$-colorings of $G$, one can obtain a
Fully Polynomial Randomized Approximation Scheme (FPRAS) for counting the number of such colorings. In statistical mechanics, sampling proper colorings is central to simulation based studies of phase transitions and correlation decay (see, e.g., the paper of Martinelli and Olivieri~\cite{martinelli1994}). 

The problem is computationally tractable if we are allowed significantly more colors than the maximum degree $\Delta$ of the graph. The more colors we are allowed, the easier it appears to be to produce a random $k$-coloring. Indeed, if $k$ is much smaller than $\Delta$, it is NP-hard to even determine whether a valid $k$-coloring exists~\cite{GAREY1976237}. 
Sampling algorithms, therefore, typically require a lower bound on $k$ in terms of $\Delta$ in order to guarantee efficiency. There has been a steady stream of works that have progressively reduced the lower bound on $k$ in terms of $\Delta$. 

Most works in this line of research focus on producing \textit{approximately} uniform samples, for approximate solutions often suffice in applications. In this setting,
the input to the problem consists of an undirected graph $G$ with $n$ vertices and maximum degree $\Delta$, a number $k$ and a parameter $\epsilon\in (0,1)$ . The goal is to generate a $k$-coloring whose distribution is within $\epsilon$ (in total variation distance) of the uniform distribution on the set of all $k$-colorings of $G$. Let $k_{+}(\Delta)$ be the smallest integer $k^*$ such that for all integers $k > k^*$\phnote{Changed defn of $k_{+}$ slightly}
there is such a sampling algorithm running in expected time $\poly(k,n, \log (1/\epsilon))$ for all $\epsilon >0$ (the subscript $+$ in $k_+$ indicates that we allow some error).
By showing that the Markov chain based on Glauber Dynamics mixes fast 
whenever $k>2\Delta$, Jerrum \cite{Jerrum} established that $k_+(\Delta) \leq 2\Delta$\footnote{Jerrum in~\cite{Jerrum} mentions that the Glauber Dynamics mixes in polynomial time even when $k=2\Delta$ and credits Frieze for this observation: hence, we have $k_+(\Delta)\leq 2\Delta-1$. }. 
Similar results appeared in the statistical physics literature (see Salas and Sokal \cite{Salas1997}); also, the path coupling approach developed by Bubley and Dyer~\cite{bubley-dyer} can be used to provide an alternative justification for Jerrum's result. Subsequent works obtained better upper bounds
on $k_+(\Delta)$. Vigoda~\cite{Vigoda-better} provided a better analysis of
Glauber dynamics (by relating it to a different Markov chain based on \emph{flip dynamics}) and concluded that $k_+(\Delta) \leq \frac{11}{6}\Delta$; recently, Chen, Delcourt, Moitra, Perarnau and Postle~\cite{chen_improved_2019} showed that $k_+(\Delta) \leq (\frac{11}{6}-\delta)\Delta$ (for a positive $\delta \sim 10^{-4}$). 
Even better upper bounds are known for certain special classes of graphs.
For graphs with girth at least $9$, Hayes and Vigoda~\cite{Hayes-Vigoda} showed that for all $\delta > 0$, we have $k_+(\Delta) \leq (1+\delta)\Delta$ provided $\Delta \geq c_\delta \ln n$ (where $c_\delta$ is a constant depending on $\delta$); for graphs of girth at least $6$ and large enough $\Delta$, Dyer, Frieze, Hayes and Vigoda~\cite{Dyer-Hayes}  showed that $k_+(\Delta) \leq 1.49\Delta$; for planar graphs Hayes, Vera and Vigoda~\cite{DBLP:journals/rsa/HayesVV15} obtained the sub-linear bound  $k_+(\Delta)\leq O(\Delta/\ln{(\Delta)})$.




\subsection{Perfect sampling} 


The algorithms described above produce samples that are only approximately uniform. The variation from uniformity can be reduced by allowing the algorithm to run longer, but it cannot be made zero; these methods do not yield perfectly uniform samples. Apart from its independent theoretical
appeal, perfect sampling has some advantages over approximate sampling.
It potentially yields FPRASs with smaller expected running time~\cite[Theorem~7]{Huber98}, because unlike with approximate sampling algorithms there is no need to ensure that the output distribution of the algorithm is sufficiently close to the target distribution in total variation distance. 
Moreover, perfect sampling algorithms are typically designed in such a way that the output produced when the algorithm stops is guaranteed to be uniform. One might be unable to formally guarantee that the expected running time is small; yet the quality of the output is never in question. In contrast, for efficient algorithms for approximate sampling, the running time may be bounded,
but in the absence of guarantees (on the mixing time, say) the output distribution may be far away from the target distribution, thereby rendering the output unreliable for statistical applications.
\sidnote{Please look at this reworded statement.}
\phnote{This para (the part in red) seems too verbose. Can it be rewritten in a couple of lines. This is not urgent and can be deferred to post-submission}


 

The intriguing fact that perfect sampling is in general possible using Markov chains was established by Propp and Wilson \cite{PW} in a seminal work which introduced the technique of \emph{coupling from the past (CFTP)} to generate perfectly uniform samples; Levin, Peres and Wilmer~\cite[Section~22.1]{LevinPeres}
point out that ideas that underlie CFTP can be traced back to the 1960s.
This paradigm has been applied to the problem of perfectly sampling $k$-colorings.
However, in contrast to the best bounds on $k_{+}(\Delta)$, which grow only linearly in $\Delta$,  the upper bounds on $k$ for perfect sampling are less impressive. \sayantannote{the red line is repeated}To better describe and compare these results, let us define $k_0(\Delta)$ to be the minimum integer $k^*$ such that there is a randomized algorithm that, given an $n$-vertex graph of maximum degree $\Delta$ and integer $k>k^*$, produces a perfectly uniform $k$-coloring of $G$ in expected time $\poly(k,n)$. 
By applying CFTP with the bounding chain approach Huber~\cite{Huber98, Huber04}, showed that  $k_0(\Delta)\leq \Delta^2+2\Delta$\footnote{Huber~\cite{Huber98} presents two different algorithms for sampling colorings which together imply that $k_0(\Delta)\leq  \min\left(\Delta^2+2\Delta,\frac{\Delta\ln n}{\ln\ln n}\right)$; however, the analysis of the algorithm which gives $k_0(\Delta)\leq  \frac{\Delta\ln n}{\ln\ln n}$ seems incomplete (the algorithm actually shows that $k_0(\Delta)\leq r\Delta$, where $r$ is the smallest natural number such that $r^r>n$; the journal version of the paper~\cite{Huber04} does not include this algorithm); an alternative algorithm based on similar ideas and  achieving the same bound is described in detail in the preliminary version of this work: see \href{https://arxiv.org/abs/1909.10323v1}{arXiv:1909.10323v1}.}.

Another paradigm for perfect sampling, related to the Moser-Tardos framework for algorithmic versions of the Lov\'asz local lemma, and also to the celebrated cycle-popping algorithm of Wilson for sampling uniformly random spanning trees, has recently been proposed by Guo, Jerrum and Liu~\cite{guo_uniform_2017}.  However, it turns out that that when this framework is applied to the problem of sampling $k$-colorings, it degenerates into usual rejection sampling: one samples a uniformly randomly coloring (including improper colorings), accepts if the coloring is proper, and rejects the current coloring and repeats otherwise.  The expected running time of such a procedure is proportional to the inverse of the fraction of proper colorings among all colorings, and hence cannot in general be bounded by a polynomial in the size of the graph. Recently, Feng, Guo and Yin \cite{feng} extending the ideas of~\cite{guo_uniform_2017} showed that Huber's result can be improved to obtain an expected polynomial time perfect sampling algorithm when $k \geq\Delta^2-\Delta+3$; their algorithm requires time $O(n\exp{(\exp{(\poly{(k)})})})$. 
Note that the upper bounds on $k_0(\Delta)$ obtained in these works is quadratic in $\Delta$ in contrast to the linear upper bounds for approximate sampling.

Hence the question remains: can $k$-colorings be efficiently and perfectly sampled when $k$ is a constant times $\Delta$? It was observed that such an improvement can be obtained if one relaxes the (expected) running time to be polynomial in only $n$ (and not in $\Delta$ and $k$). To state these results, let us define the relaxed version $\widetilde{k_0}(\Delta)$ as follows: $\widetilde{k_0}(\Delta)$ is the minimum integer $k^*$ such that there is a randomized algorithm that, given an $n$-vertex graph of maximum degree $\Delta$ and integer $k>k^*$, produces a perfectly uniform proper $k$-coloring of $G$ in expected time $\poly_{\Delta,k}(n)$ (i.e., the dependence on $n$ is polynomial, but the dependence on $\Delta$ and $k$ can be arbitrary). 
A general method for perfect sampling based on approximate counting was suggested by Jerrum, Valiant and Vazirani~\cite[Thm 3.3]{JerrumValiantVazirani}; that is, using an efficient algorithm for deterministically approximately counting the number of $k$-colorings, one can efficiently sample perfectly. This approach when used together with 
the deterministic approximate counting algorithm of Gamarnik and Katz \cite{GamarnikKatz},
yields $\widetilde{k_0}(\Delta) \leq 2.78 \Delta$ for triangle free graphs; approximate
counting algorithms in subsequent works due to Lu \& Yin \cite{LuYin},
and Liu, Sinclair \& Srivastava \cite{SrivastavaSinclairLiu}, yield 
$\widetilde{k_0}(\Delta) \leq 2.58 \Delta$ and $\widetilde{k_0}(\Delta) \leq 2 \Delta$, respectively. The running time of these algorithms has the form $O(n^{f(k,\Delta)})$ (for instance in~\cite{SrivastavaSinclairLiu} the exponent contains an  $\exp{(\Delta)}$ term). Another point to be noted is that these deterministic approximate counters are based on decay of correlations or the so-called `polynomial interpolation' method of Barvinok~\cite{Barvinok}, which are not as simple as the Markov chain based algorithms (e.g., the CFTP based algorithms of Huber and in this paper).\sidnote{want to mention that MC algorithms are simpler than the correlation decay/polynomial interpolation methods}
  For instance, the MC based algorithms offer an appealing combinatorial explanation of the number of colors required for `mixing' to take place unlike the other methods. In any case, none of these algorithms yield truly linear bounds on $k_0(\Delta)$.  
  

\subsection{Our contribution}




\begin{theorem}[Main result]
\label{thm:main}
$k_0(\Delta) \leq 3\Delta$. 
In particular, there is a randomized algorithm which we call \textup{\textsc{PerfectSampler} (\cref{alg:1})}, based on \textup{CFTP} that\sidnote{put final algo name here} 
given an $n$-vertex graph $G=(V,E)$ of maximum degree $\Delta$ and $k> 3\Delta$, returns a uniformly random $k$-coloring of $G$. The algorithm uses fair independent unbiased coin tosses (with probability $1/2$ for head and tail), and stops in expected time $O((n\log^2 n)\cdot (\Delta^2\log\Delta\log k))$. 
\end{theorem}

Our result is based on Coupling From the Past (CFTP). In the rest of this section, we briefly review CFTP as it is applied to the problem of $k$-coloring,\sayantannote{coloring or colorings?} and describe its efficient implementation using the Bounding Chain method roughly along the lines of Huber~\cite{Huber98,Huber04}. We then describe the key ideas that allow us to improve the upper bound on $k_0(\Delta)$ to $3\Delta$.

Consider the standard Markov chain for $k$-coloring that evolves based on the Glauber Dynamics: in each step a random vertex $v$ is chosen and its color is replaced by a uniformly chosen color not currently used by any of its neighbors. The standard CFTP algorithm~\cite{PW} based on this Markov chain, assumes that we generate a sequence of random variables, $(v_{-1},\sigma_{-1}), (v_{-2},\sigma_{-2}), \ldots,$
where $v_i$ is a random vertex in $V(G)$ and $\sigma_i$ is a random permutation of the set of colors $[k]$, chosen uniformly and independently. 
For $i=-1,-2,\ldots$, let $U_i$ be the operation on $k$-colorings
that performs the following update based on the pair $(v_i,\sigma_i)$.
Given a proper $k$-coloring $\chi:V \rightarrow [k]$ and the pair $(v_i,\sigma_i)$, let $U_i(\chi)$ be the coloring $\chi'$ defined as follows: $\chi'(v_i)$ is the first color in $\sigma_i$ that is not in $\chi(N(v_i))$ and for vertices $w\neq v_i$, $\chi'(w)=\chi(w)$.  Note that $U_i$ maps proper $k$-colorings to proper $k$-colorings. The CFTP algorithm is based on the following principle.
Let $t$ be an integer such that $U_{-1} \circ U_{-2} \circ \cdots \circ U_{t}$ is a constant function, that is, this sequence of updates applied to every proper $k$-coloring results in the same coloring, say $\chi_{\out}$. The algorithm outputs $\chi_{\out}$. Note that this output does not depend on the choice of $t$.  For example, we could run through $i=-1,-2,\ldots$ until the first index $t$ when $U_{-1} \circ U_{-2} \circ \cdots \circ U_{t}$ becomes a constant function, and output the unique $k$-coloring in its image.  
It is well known that if $k> \Delta+1$, then with probability $1$ such a $t<\infty$ exists, and $\chi_{\out}$ is uniformly distributed in the set of all colorings.

The randomized algorithm as stated above is not efficient. The number of starting states for the chains grows exponentially with $n$, and the time taken to keep track of the updates  on each will be prohibitively large. To keep the computation tractable, Huber employs the Bounding Chain (BC) Method (pioneered by him in the context of coloring and also independently by  H\"aggstr\"om \& Nelander~\cite{Hagg99}).
In the Bounding Chain method, instead of precisely keeping track of the various states that are reached after each update operation, we maintain an upper bound: a list of states, which contains all the states that could potentially be reached. In fact, this upper bound for the state reached after update $U_{j-1}$ has been applied will have the special form: $\prod_{v \in G} L_j(v)$, where $L_j(v) \subseteq [k]$. 
That is, when simulating the actions of successive updates $U_{j-1},\ldots, U_{t}$, we do not explicitly maintain the colors across vertices, but rather just a list $L_j(v)$ that includes all colors that vertex $v$ can take in any $k$-coloring reached by performing these updates starting from any initial $k$-coloring. Note in particular, that our choice of $t$ will be good if we can ensure that $|L_0(v)|=1$ for all vertices $v$; for, then we know that $U_{-1} \circ U_{-2} \circ \cdots \circ U_{t}$ is a constant function (on the space of $k$-colorings), and $\chi_{\out}$ is the unique coloring in $\prod_v L_0(v)$. 

We are now in a position to give a high-level description of Huber's BC method~\cite{Huber98,Huber04}. After a short initial phase of updates which act as warm-up, Huber maintains the invariant that $|L_j(v)|\leq \Delta+1$ for all vertices $v$. To measure progress towards the goal that $|L_0(v)|=1$ for all vertices $v$, let us define $W_j$ to be the number of vertices $v$ such that $|L_j(v)|=1$. Hence, we want $W_0=n$. Now, suppose that at time $j$ we have the update operation $U_j$ given by $(v_j,\sigma_j)$.
Consider $L_j$ such that $|L_j(v)|\leq \Delta+1$ and let $S_{L_j}(v)$ be the union of colors present in the lists of neighbors of $v$.
Then, Huber sets $L_{j+1}(v)=\set{\sigma(1)}$ if $\sigma(1)\notin S_{L_j}(v)$; otherwise $L_{j+1}(v)=\set{\sigma(1),\ldots,\sigma(\Delta+1)}$. For all $w\neq v$, $L_{j+1}(w)=L_{j}(w)$.
Notice that if $\chi \in \prod_{w \in G} L_j(w)$ then $U_j(\chi)\in \prod_{w \in G} L_{j+1}(w)$ as $v$ definitely finds an available color in $\set{\sigma(1),\ldots,\sigma(\Delta+1)}$ and hence $(U_{j}(\chi))(v)\in \set{\sigma(1),\ldots,\sigma(\Delta+1)}$.
Hence, we make progress (towards our goal of $W_0=n$) whenever $\sigma(1)\notin S_{L_{j}}(v)$ and suffer a loss otherwise. To be able to have a non-trivial probability of making progress we need that $k>|S_{L_{j}}(v)|$ ($|S_{L_{j}}(v)|$ can potentially be as large as $\sum_{w\in N(v_j)}|L_j(w)|$ which in turn can be $\Delta\times(\Delta+1)$) and this is ensured by having $k>\Delta^2+\Delta$.
However, $W_j$ evolves as a random walk on $\set{0,\ldots,n}$ (with $n$ as absorbing state) and to have sufficient drift to the right we require an extra margin of $\Delta$ in $k$ and hence Huber assumes $k>\Delta^2+2\Delta$.
It then follows that in expected time $\poly(n,k)$ one can find the starting time $t$ so that $|L_0(v)|=1$ for all vertices $v$ and hence $U_{-1} \circ U_{-2} \circ \cdots \circ U_{t}$ is a constant function. 
We omit the detailed analysis of Huber's method, but note that for this method to succeed, $k$ must be larger than the product of the maximum degree ($\Delta$) and the upper bound on the size of $L_j(v)$ that we can ensure plus an extra margin of $\Delta$; this implementation, therefore, yields only a quadratic bound on $k_0(\Delta)$. 

\sidnote{@jaikumar: would you like to mention that random walk on a line has been used in the past in the context of sampling here? and are there any references?}

We improve upon this using a better implementation of the Bounding Chain Method. 
After a short initial warm-up phase of updates (which we formally call the collapsing phase), the set of colors $L_i(v)$ in our implementation will be of size at most two; this will allow us to ensure perfect sampling as long as $k>3\Delta$. The correctness of our algorithm will still rely on the Markov chain based on Glauber Dynamics described earlier. However, we depart substantially from earlier works in designing our updates that implement the Glauber Dynamics. Recall that a sequence of random update operations $U_{-1}, U_{-2}, \ldots,$ need to be designed in the CFTP algorithm. In Huber's approach the $U_i$'s were independently and identically distributed (based on independent choices of the pairs $(v_i,\sigma_i)$). Our new update operations will not be chosen independently but will have a rather special distribution. This distribution is designed keeping in view our goal of restricting the bounding lists $L_t(v)$ to size at most two, and is best understood in the context of the evolution of these lists in our implementation of the Bounding Chain Method, which we describe in the subsequent sections
. For now, we outline the main properties of this distribution. For $0> i > j$, let $U[i,j] := (U_{i}, U_{i-1}, \ldots, U_{j})$ and let
$U(i,j) := U_i \circ U_{i-1}\circ \cdots \circ U_j$.  (Note $U[i,j]$ refers to the sequence (or array) of random choices that describe the $i-j+1$ update functions while $U(i,j)$ refers to the composed update function.)
\begin{lemma} \label{lm:main}
Let $G$ be an $n$-vertex graph with maximum degree $\Delta$. Let $k> 3\Delta$.
Then, there is a positive integer $T$ which is $\poly(k,n)$, a joint distribution $\mathcal{D}$ for $T$ updates, $U[-1,-T]$, 
and a predicate $\Phi$ on the support of $\mathcal{D}$, satisfying the following conditions.
$(T=2n\ln n{(k-\Delta)}/({k-3\Delta}) + |E(G)| + n$ where $|E(G)|$ is the number of edges in $G$.)
\begin{enumerate}[label=(\alph*)]
    \item A sample with distribution $\mathcal{D}$ can be generated and
          the predicate $\Phi$ can be computed in time $\poly(n,k)$; further, each update instruction $U_i$ is efficient, i.e., given a $k$-coloring $\chi$, $U_i(\chi)$ can be computed in time $\poly{(k,n)}$;
    \item If $Z$ is a uniformly generated proper $k$-coloring of $G$ and 
    $U[-1,-T]$ is picked according to $\mathcal{D}$ independently of $Z$, then $U(-1,-T)(Z)$ is a uniformly distributed $k$-coloring;
    \item  If $\Phi(U[-1,-T])=\textup{\textsc{true}}$, then $U(-1,-T)$ is a constant function (on the set of proper $k$-colorings), that is, \\  $|\{U(-1,-T)(\chi): \mbox{$\chi$ is a $k$-coloring}\}|=1$;
    \phnote{removed remark on efficiency as it is implied by the furthermore in (a)}
    \item $\Pr_{\mathcal{D}}[\Phi(U[-1,-T])=\textup{\textsc{true}}] \geq \frac{1}{2}$.
\end{enumerate}
\end{lemma}
\paragraph{Remark:} The update operations (which act on an exponentially large set) need to be represented succinctly for our algorithm to be efficient. Each operation will be encoded succinctly by tuples. 
(For example, in the discussion above the tuple $(v_i,\sigma_i)$ can be thought of as the encoding of the update operation $U_i$.) The encoding we use is described below. Thus, in part \cref{lm:main} (a), when we need to generate a sample from $\mathcal{D}$, we actually generate the sequence of $T$ tuples corresponding to the update operations. Similarly, the predicate $\Phi$ is expected to take as argument a sequence of tuples and efficiently evaluate to $\textsc{true}$ or $\textsc{false}$; further when $\Phi=\textsc{true}$ we can efficiently compute the (unique) image of $U(-1,-T)$ from the tuples.
We will ensure that the decoding is efficient: given a tuple that represents an update operation $U$ and a proper $k$-coloring $\chi$, the coloring $U(\chi)$ can be computed efficiently.


We then have the following natural randomized algorithm for perfectly sampling $k$-colorings.

\begin{algorithm}
	    \For{$i=0,1,2, \ldots,$} {
            Generate $U[-iT-1,-(i+1)T]$ according to $\mathcal{D}$ \;
            \If{$\Phi(U[-iT-1, -(i+1)T])=\textup{\textsc{true}}$}{
                     Output the unique $k$-coloring in the image of $U(-1,-(i+1)T)$ and \textbf{STOP}; \label{line:step4}
            }
        }
    \caption{\textsc{PerfectSampler}}\label{alg:1}
\end{algorithm}

\begin{proof}[Proof of \cref{thm:main}]
We wish to show that the output of the above algorithm is uniformly distributed in the set of all $k$-colorings. Let $U[-1,-T], U[-T-1, -2T], \ldots, U[-(i-1)T, -iT], \ldots$ be the random sequences that arise when the algorithm 
samples independently from the distribution $\mathcal{D}$. It may be that some of the later sequences are not used by the algorithm if the predicate $\Phi$ evaluates to true on an earlier sequence, but we define all of them anyway for our argument. Let $\chi$ be uniformly chosen $k$-coloring. Fix $i\geq 1$. Then, by \cref{lm:main} (b), $\chi_i=U(-1,-iT)(\chi)$ is uniformly distributed. Let $\chi^*$ be the output of the above algorithm. By \cref{lm:main} (c),
$\chi^*$ and $\chi_i$ are identical whenever $\Phi$ evaluates to true on one of $U[-1,-T], U[-T-1, -2T], \ldots, U[-(i-1)T-1, -iT]$, which happens with probability at least $1-2^{-i}$ by \cref{lm:main} (d). From the duality of total variation distance and coupling, it follows that the distribution of $\chi^*$ is within $2^{-i}$ of the uniform distribution (the distribution of $\chi_i$). Since, $i$ was arbitrary, we see that the distribution of $\chi^*$ is uniform. 

The algorithm is efficient\footnote{\cref{line:step4} of the algorithm can be performed by taking the trivial coloring $\chi=1^V$ and then outputting $U(-1,-(i+1)T)(\chi)$; however, in our implementation of the updates the condition $\Phi(U[-iT-1, -(i+1)T])=\textsc{true}$  will be validated by producing the unique coloring $\chi$ in the image of $U(-iT-1, -(i+1)T)$; so for \cref{line:step4} we output $U(-1, -iT)(\chi)$.  }
because it performs at most two iterations of the \emph{for} loop in expectation, and sampling from $\mathcal{D}$ and the computation of $\Phi$ are efficient by part \cref{lm:main} $(a)$. 
For a detailed analysis of the running time refer to \cref{subsec:2.5}.
\end{proof}

In other words, let $i$ be the first index in \cref{alg:1} such that we find $\Phi(U(-iT-1,-(i+1)T))=\textsc{true}$. So, we know that $|\{U(-1,-T)(\chi): \mbox{$\chi$ is a $k$-coloring}\}|=1$. Now, we update this unique coloring with $U(-iT,-1)$ and output the updated coloring.
The correctness of the algorithm and that it runs in expected time $\poly(n,k)$ follow immediately from \cref{lm:main}; in particular, this justifies \cref{thm:main} barring the expected running time.  
In the rest of this introduction, we describe the distribution $\mathcal{D}$ and outline our proof of the lemma


\paragraph{Representation of update operations:} Our approach is inspired by the Bounding Chain method. To make this precise, we need a definition. By a \emph{bounding list} we mean a list of the form $L=(L(v): v \in L)$, where each $L(v)$ is a set of colors. We refer to $L(v)$ as $v$'s list of colors; thus $L$ is a list of lists. We say that a $k$-coloring $\chi$ is compatible with $L$, and write $\chi \sim L$, if $\chi(v) \in L(v)$ for all $v$, that is, if $\chi \in \prod_v L(v)$. We are now in a position to describe the representation we use. Each update operation will be associated with a $5$-tuple of the form $\alpha=(v,\tau, L, L', M)$, where $v$ is a vertex, $\tau \in [0,1]$, and $L$ and $L'$ are bounding lists, and $M$ is a sequence of at most $\Delta+1$ distinct colors. We refer to the update operation associated with $\alpha$ as $U_\alpha$. 
Thus, the distribution of $U[-1,-T]$ will be specified by providing a randomized algorithm for generating the corresponding sequence of tuples $\alpha[-1,-T]$ and letting $U_t=U_{\alpha_t}$. We now describe some of the important features of this sequence.

Fix $t \in \{-T, \ldots, -1\}$.
Suppose $\alpha[t-1,-T]$ have been generated. Now, consider $\alpha_{t} = (v_t, \tau_t, L_t, L'_t, M_t)$. We will ensure that the following conditions hold.
\begin{enumerate}[label=(\alph*)]
    \item[{[C1]}] The random vertex $v_{t}$ is independent of $\alpha[t-1,-T]$. In Huber's chain, $v_t$ was actually uniformly distributed; we will not be able
    to ensure that; in fact, some of our vertices will be determined by the index $t$ (the current time step); for example, $v_{-T}$ will be a fixed vertex of the graph, not a random vertex.
    \item[{[C2]}] The distribution of $\alpha_t$ will implement the Glauber Dynamics at vertex $v_t$ in the following sense. 
    Condition on $\alpha[t-1,-T]$ and $v_t$ (the first component of $\alpha_t$). Fix a coloring $\chi$ in the image of $U(-T,t-1)$ (note that this operator is determined completely by 
    $\alpha[t-1,-T]$, which we have conditioned on). Now, we require that 
    $\chi'=U_t(\chi)$ has the following distribution: $\chi'(w)=\chi(w)$, for all $w \neq v_t$ and $\chi'(v_t)$ is uniformly distributed in the set of colors
    $[k]\setminus \chi(N(v_t))$. If this condition is satisfied, then we say that $\alpha_t$ satisfies $\gd{(\chi,v_t)}$. Note that this will ensure \cref{lm:main} (b).
    
    \item[{[C3]}] 
    The lists $L_t$ impose a certain restriction on the domain of $U_t$: $U_t$ will be defined only on colorings $\chi \sim L_t$. Thus $L_t$ represents a \emph{precondition} for $U_t$ to be applicable. Similarly, $L'_t$ represents a \emph{postcondition}: if $\chi \sim L_t$, then $U_t(\chi) \sim L'_t$. We will, therefore, have in our sequence that $L'_t = L_{t+1}$.
    Also, $L_{-T}$ will be $([k])^V$. If the above discipline concerning preconditions and post-conditions is maintained, then for the image of $U(-T,-1)$ to be a singleton, it is enough that $|L'_{-1}(v)|=|L_0(v)|=1$ for all $v \in V$. Indeed, our predicate $\Phi$ will 
    verify this by examining $\alpha_{-1}$; to establish \cref{lm:main}(d), we will show that this condition holds with probability at least $\frac{1}{2}$. 
\end{enumerate}

\paragraph{The key ideas:} We discussed above some of the conditions that our random sequence of tuples $\alpha[-1,-T]$ will satisfy. We now informally
describe how $\alpha_t$ is translated or decoded to obtain $U_t$ and how $\alpha[-1,-T]$ is generated. This informal description will differ slightly from the more formal one we present in \cref{sec:def}; but it will let us motivate our definitions, and also throw light on how the new method makes do with fewer colors than Huber's method. 

Initially, at time $-T$, each vertex's list is $[k]$: that is, $L_{-T}(v)=[k]$ for all all $v$. We wish to ensure that in the end all lists have size $1$: that is, $|L'_{-1}(v)|=1$ for all $v$. We will achieve this in two phases. At the end of the first phase, we will ensure that all vertices have lists of size at most $2$ with probability 1. We refer to this phase as the \emph{collapse} phase.  The second phase, the \emph{coalesce} phase, will ensure that with  probability at least $1/2$, the lists of all vertices have size one. The total number of updates in the first and second phases put together will be $T$. We now briefly describe the ideas involved in the two phases.

The updates in these two phases will be generated by two primitives. (i) The first primitive takes vertex $w$ and a set $A$ of at most $\Delta$ colors and produces an update called \emph{compress} update; after this update, the list at $w$ will have at most one element outside $A$. (ii) The second primitive takes a vertex $v$ and generates a random update called \emph{contract} update;
for this primitive to be used, we must ensure that the previous updates have spruced up the neighborhood of $v$ which is said to have occurred when the union of colors in the lists of neighbors of $v$ has size less than $k-\Delta$. But whenever such an update is performed, the list of $v$ immediately contracts to size at most two; in fact, with significant probability it contracts to size one. We describe these primitives in detail in the following sections. For now, let us see roughly see how they are deployed to achieve the goals of the two phases.

\newcommand{\spruceup}{\textsc{spruceup}}
\newcommand{\contract}{\textsc{contract}}

\textbf{Collapsing phase:} The reduction in list size all the way to just two will be achieved by using {contracting} updates. However, for such an update to be applied at a vertex, the total number of colors in the union of the lists of its neighbors must be small (for us it will need to be less than $k-\Delta$; in fact, we will ensure that it is at most $2\Delta$).  Note that our initial bounding list $L_{-T}$ does not satisfy this condition; all lists have size $k$. We, therefore, need to first \emph{spruce up} the neighborhood. Fix an ordering of the vertices, say $v_1,v_2,\ldots,v_n$\footnote{ There is a notational overload here: earlier we had used $v_i$ to denote the vertex chosen at time step $i$ for the update operation, but now we mean it to be the $i^{th}$ vertex in the ordering. This will be clear from the context.}. Conceptually, the contracting phase will
perform the actions in the following sequence:
\[ \spruceup(v_1), \contract(v_1), \spruceup(v_2), \contract(v_2), \]\begin{align*}
\ldots, \spruceup(v_n), \contract(v_n).\end{align*}
Here $\spruceup(v_i)$ is a composite update operation.
It consists of several updates that compress the lists at the neighbors of $v_i$ using a common set $A_i$ of $\Delta$ colors. For example, if $v_1$ has $d_1$ neighbors, then $\spruceup(v_1)$ will consist of $d_1$ compress update operations, one for each of its neighbors. It is easy to see that then the union of the lists at $v_1$'s neighbors will have at most $2\Delta$ colors (each of the at most $\Delta$ neighbors will contribute at most one new color outside $A_1$)---the neighborhood of $v_1$ is thus spruced up. In general, for $v_i$ the operation $\spruceup(v_i)$ will perform the compress operation on those neighbors of $v_i$ which are after $v_i$ in the ordering.
Once the neighborhood of $v_i$ has been spruced up in this fashion, a single contract update ensures that the list size of $v_i$ contracts to two. There is one subtlety, however. After contracting the lists of $v_1,\ldots,v_i$, when we proceed to spruce up the neighborhood of $v_{i+1}$, we only perturb the lists of vertices after $v_{i+1}$ (in particular, the lists of vertices before $v_{i+1}$ remain unperturbed): yet we need to ensure that the union of the lists at $v_{i+1}$'s neighbors (both preceding and succeeding) will have at most $2\Delta$ colors. So we choose the set $A_{i+1}$ so that it includes at least one color from the lists of the neighbors where a contraction has already been achieved. In ~\cref{collapsephase}, we describe the collapsing phase in detail.

\textbf{Coalescing phase:} Suppose the collapsing phase has successfully contracted all lists to size at most two. Our goal now is to extend the above sequence with some more randomly generated updates so that with probability at least $1/2$ the final list sizes all become one. We again use the contract update operation described above, this time exploiting the feature that it 
contracts lists to size just one with significant probability. 
However, while vertices with list size two can hope to see a reduction in their list size, a vertex whose list size is already one will, with some probability, acquire a list size of two. In this phase, we randomly pick vertices and perform a contract update on them. 
Note that a contract update never results in a list of size more than two; so, all neighborhoods stay spruced up at every point in the coalescing phase. If we track the number $W_t$, which is the number of vertices with list size $1$ at time $t$,  this quantity performs a random walk on the number line (between $0$ and $n$, with $n$ as absorbing)\footnote{Whenever for a vertex $v$ all its neighbors have list size $1$ then the contract update applied to $v$ produces a list of size $1$ at $v$.} with a non-negligible bias towards $n$. We observe that if $k$ is large enough ($k > 3\Delta$), then with high probability this walk will hit $n$ within $\poly(n,k)$ steps, and helps us justify \cref{lm:main} (d). In ~\cref{subsec:2.3}, we describe the coalescing phase in detail.

\subsection*{Organisation of this paper}
In the following sections, we elaborate on ideas outlined above, and justify \cref{lm:main}. In \cref{subsec:2.1}-\cref{subsec:2.3}, we formally define $T$, the distribution $\mathcal{D}$, the primitives that we use to generate the $\alpha$s at different stages of the algorithm, and the precise correspondence between the strings of type $\alpha$ and the corresponding update operations of type $U_{\alpha}$. Finally, in \cref{subsec:2.5} we formally define the predicate $\Phi$ and collate all our results from the previous sections to establish parts $(a), (b), (c)$ and $(d)$ of \cref{lm:main}. The running time analysis of our algorithm is also presented in \cref{subsec:2.5}.

\section{The distribution $\mathcal{D}$ and the predicate $\Phi$}
\label{sec:def}
In this section, we will prove \cref{lm:main}. Recall that we have a graph $G=(V,E)$ on $n$ vertices and the number of colors $k>3\Delta$.
The update sequence $(\alpha_{-T},\ldots,\alpha_{-T'-1})$ will correspond to the collapsing phase of our algorithm and $(\alpha_{-T'},\ldots,\alpha_{-1})$ to the coalescing phase. In particular, we set $T'=2\frac{k-\Delta}{k-3\Delta}n\ln n$ and $T=T'+|E(G)|+n$ where $|E(G)|$ is the number of edges in $G$.
The reasons for these values will be clear in the following subsections.
\subsection{The update $\alpha$ and its relation to $U_\alpha$}
\label{subsec:2.1}
Recall that an update operation is represented by a tuple $\alpha$ of the form $(v,\tau,L,L',M)$. In the previous section, we informally indicated the role played by each of the components of this $5$-tuple. In this section, we specify exactly how these components are generated and how they determine the update operation $U_{\alpha}$. 
As stated in the introduction, we have two types of updates, the compress update and the contract update. The generation and decoding methods are
different for the two. We describe, for each type, how
the corresponding $\alpha$ is generated and how, given such an $\alpha$, the corresponding $U_\alpha$ is applied to a coloring $\chi$. Our final sequence of updates will be obtained by generating the updates one after another according to a strategy that we describe later.

The update operation associated with $\alpha=(v,\tau,L,L',M)$ will act on colorings 
$\chi \sim L$; that is, whenever we use $\alpha$ in our sequence, it will be
guaranteed that the previous update operations result in a coloring $\chi \sim L$. 
However, if each $L(v)=[k]$ for all $v$, then $U_\alpha$ acts on all colorings.
Fix a coloring $\chi$. The operation $U_\alpha$ will attempt to recolor the vertex $v$ (leaving the colors of the other vertices unchanged) by picking a color from the sequence $M$. 
The $\alpha$ we generated will have $L'(v) = M$ barring the order; this will ensure that $U_\alpha(\chi)\sim L'$.
In order to ensure that the coloring is proper, the color chosen for $v$ must avoid the colors used by $v$'s neighbors. In particular, if $|L(w)|=1$ for a neighbor $w$ of $v$, then the unique color in $L(w)$ will never be a candidate color for $v$. Thus, the following two sets will play a central role in our definition of $U_{\alpha}$:
\begin{align*}
    S_{L}{(v)}&=\bigcup\limits_{w\in N(v)}L(w)& &\mbox{and}&
    Q_{L}{(v)}&=\bigcup\limits_{\substack{w\in N(v) \\ \abs{L(w)}=1}}L(w).
\end{align*}
In the following subsections we will consider $\alpha$s of two types, depending on the size of $M$.
\begin{description}
\item[Type \underline{compress} ($|M|=\Delta+1$):] Such an $\alpha$ will be used to \\
spruce up the neighborhoods.  
\item[Type \underline{contract} ($|M| \leq 2$):] Such an $\alpha$ will be used in the collapsing phase to contract the list sizes to size at most two, and again in the coalescing phase to make make all list sizes $1$.
\end{description}


\subsubsection{Compress updates}
\label{subsec:WarmUpUpdates}
\newcommand{\init}{\textup{\textsc{in}}}
\newcommand{\WarmupGen}{\textup{\textsc{Compress.gen}}}
\newcommand{\WarmupDecode}{\textup{\textsc{Compress.decode}}}
To specify the compress updates we will present two procedures: $\WarmupGen{}$ and \\ $\WarmupDecode{}$. The procedure $\WarmupGen{}$ takes a tuple  $\alpha_{\init}=(v_{\init},\tau_\init,L_\init, L'_\init, M_\init)$, a vertex $v$ and a list $A$ consisting of $\Delta$ colors, and returns another tuple.
This procedure is randomized: its output $\alpha_{\out}$ is a random tuple of type {compress}, and is the immediate successor of $\alpha_\init$ in our sequence of updates. The update operation corresponding to such a tuple is obtained using procedure $\WarmupDecode{}$, which takes a tuple $\alpha_\out$ (produced by $\WarmupGen{}$) and a coloring $\chi\sim L'_\init$, and produces another coloring, say $\chi'$. Thus, the update operation $U_{\alpha_\out}$ is the map $\chi \mapsto \WarmupDecode[\alpha_\out, \chi]$. The following lemma describes the relationship between the two procedures, and their important properties.

\begin{lemma} \label{lm:NewWarmUpSampler}
Let $\alpha_{\init}=(v_{\init},\tau_\init,L_\init,L'_\init, M_\init)$ be an arbitrary 5-tuple, $v \in V$ and $A$ be a subset of $\Delta$ colors. Then,
\begin{enumerate}[label=(\alph*)]
\item If $\alpha_{\out}=(v_{\out},\tau_{\out},L_{\out}, L'_{\out},M_{\out})$ is a random tuple produced by $\WarmupGen[\alpha_{\init},v,A]$, then $L_{\out}=L'_{\init}$, $L'_{\out}(u)=L_{\init}(u)$ for all $u \neq v$, and $L'_{\out}(v)$ has the form $A \cup \{c\}$ for some color $c$ outside $A$.
\item For all $\chi\sim L_{\out}$, we have $\chi':=\WarmupDecode[\alpha_{\out},\chi] \sim L'_{\out}$ (with probability $1$).
\item For all $\chi\sim L_{\out}$, the coloring $\chi'$ has the same distribution as $\gd(\chi,v)$, that is, $\chi'(w)=\chi(w)$, for all $w \neq v$, and $\chi'(v)$ is uniformly distributed\footnote{Note that 
the randomness in $\chi'(v)$ arises from the random choices made in generating $\alpha_{\out}$ using ~$\WarmupGen[\alpha_{\init},v,A]$.}
in the set of colors $[k]\setminus \chi(N(v))$.
\item Except for copying of the list $L'_\init$, the expected running time of $\WarmupGen{}$ is $O(\Delta\log k+\log n)$. The time needed to update a $k$-coloring $\chi$ using $\WarmupDecode{}$ is $O(\Delta (\log \Delta \log k + \log n))$.
\end{enumerate}
\end{lemma}

To prove this lemma, we need to specify $\WarmupGen$ and $\WarmupDecode$.
Before presenting the code and the proof of the lemma, we present the idea behind them.
Given $\alpha_\init$, $v$ and $A$, we somehow want to update the color of vertex $v$. The precise color to assign to $v$ will need to depend on the current coloring $\chi$, in particular, on $\chi(N(v))$. If all we wanted was to restrict the size $L'_{\out}(v)$, we could just insist that $v$'s color be confined to a random subset of size $\Delta+1$; that is, 
\WarmupGen{} would specify a random sequence of $\Delta+1$ distinct colors and once $\chi$ is known, we would replace $\chi(v)$ by the first color in this list not currently used by any neighbor of $v$. However, as explained in the introduction, we wish to ensure that the lists of different vertices overlap  with $A$. So we actually generate a random permutation of the input set $A$, say $\sigma$ and append to it at the end a random color $c_1$ chosen from $[k]\setminus A$; thus $\alpha_{\out}$ has the form $(v_\out,\tau_\out,L_\out,L'_\out, (\sigma,c_1))$;
here $\tau_\out$ will be chosen uniformly from $[0,1]$; its role will become clear soon. This simple procedure is our $\WarmupGen{}$. Now, once such an $\alpha_\out$ has been specified, to update $\chi(v)$, we have a choice: either we pick $c_1$ or one of the colors from $A$. If $c_1$ is an invalid option (it is being used by a neighbor of $v$), then we have no choice but to pick a color from $A$ (there must be one available!). Now, $c_1$ will be a valid option with probability $(k-|\chi(N(v))\cup A|)/(k-\Delta)$, whereas such a color should actually be used to replace $\chi(v)$  with probability $(k-|\chi(N(v)) \cup A|)/(k-|\chi(N(v))|)$. So whenever $c_1$ is a valid option, we replace $\chi(v)$ by $c_1$ with probability $(k-\Delta)/(k-|\chi(N(v))|)$ and with the remaining probability we use the first valid color from $\sigma$. To implement this acceptance sampling we pick a random number $\tau_\out \in [0,1]$ and accept $c_1$ if it is at least the threshold $1-(k-\Delta)/(k-|\chi(N(v))|)$. This is all that $\WarmupDecode{}$ does. We now present the code (which may be skipped) that implements what we discussed above and formally prove  Lemma~\ref{lm:NewWarmUpSampler}. 

%
%
\begin{algorithm}
\caption{\textsc{Compress}: generation and decoding}
\label{alg:WarmUpSampler}
\setstretch{1.3}
\SetKwInOut{KwIn}{Input}
\SetKwInOut{KwOut}{Output}
\SetKwFunction{gen}{gen}
\SetKwProg{myproc}{Function}{}{}
\myproc{\gen{}\rm{:}}{
 
 \KwIn{$\alpha_{\init}=(v_{\init},\tau_{\init},L_{\init},L'_{\init}, M_{\init})$, $v\in V$ and $A\subseteq [k]$ with $|A|=\Delta$}
 
\KwOut{$\alpha_{\out}=(v_{\out},\tau_{\out},L_{\out},L'_{\out},M_{\out})$}

$\tau_{\out} \xleftarrow{R}[0,1]$ ;~
$\sigma \xleftarrow{R} S_{A}$;~
$c_1\xleftarrow{R}[k]\setminus A$ \; $L'_{\out} \gets L'_{\init}$ \;\label{bigstep} 
$L'_{\out}(v)\gets A\cup \brak{c_1}$;~ 
$M_{\out} \gets (\sigma,c_1)$ \tcp*{Appending $c_1$ at the end of $\sigma$}
\KwRet{$\alpha_{\out}=(v,\tau_\out,L'_{\init},L'_\out,M_\out)$}
}

\SetKwFunction{decode}{decode}
\myproc{\decode{}\rm{:}}{
\KwIn{$\alpha=(v,\tau,L,L',M)$ and a coloring $\chi \sim L$}
\KwOut{$\chi' \sim L'$}\label{chistep}
$\chi' \gets \chi$\;
$p_{\chi}(v) \gets  1-\frac{k-\Delta}{k-|\chi(N(v))|}$\label{c1AndPrtau>p_chi}\;
\uIf{$c_1 \notin \chi(N(v))$ and $\tau \geq p_{\chi}(v)$ \label{loop:compressdecodestart}}
   {$\chi'(v)\gets M[\Delta+1]$ \tcp*[r]{$M$ has the form $(\sigma,c_1)$ where  $\sigma$ is list of $\Delta$ colors.} } 

\Else{
    $\chi'(v)\gets$ first color in the list $M[1,\Delta]$ that is not in $\chi(N(v))$
    \tcp*{If $c_1\in \chi(N(v))$ \\ then such a color is always available as $|\chi( N(v))|\leq \Delta$.}
    }\label{loop:compressdecodeend} 
\KwRet{$\chi'$}}

\end{algorithm}

\begin{proof}[Proof of \cref{lm:NewWarmUpSampler}]

Part $(a)$ is clear from the procedure of \WarmupGen. We update $L_{\out}$ as $L'_{\init}$ and $L'_{\out}$ differs from $L'_{\init}$ only at the vertex $v$ where $L'_{\out}(v)=A\cup \brak{c_1}$.

For part $(b)$ consider any $\chi \sim L'_{\init}=L_{\out}$. Note that during \WarmupGen{} we set $L'_{\out}(v)=A\cup \brak{c_1}$ and $M_{\out}$ as $(\sigma,c_1)$ and during \WarmupDecode{} we update the color of $\chi'(v)$ from within $M_{\out}$ and for all $w\neq v$ we copy the color of $\chi$. This proves part $(b)$.


For part $(c)$ we remind ourselves that the process \gd{($\chi,v$)} requires $\chi'(v)$ to be uniformly distributed on the set $[k]\setminus \chi(N(v))$.
Notice that $c_1=M[\Delta+1]$ is a uniformly random choice of color over $[k]\setminus A$ and hence whenever $c_1$ is chosen for $\chi'(v)$ by \WarmupDecode{} we know its distribution will be uniform over $k\setminus (A\cup \chi(N(v)))$.
Also, whenever we choose a color from $M[1,\Delta]=\sigma$ in \WarmupDecode, where $\sigma$ is a uniformly random permutation of $A$, to update at $\chi'(v)$ we know that its distribution is uniform over $A\setminus \chi(N(v))$.  
Hence, to prove part $(c)$, it suffices to show that $c_1$ is chosen with probability $\frac{k-|A\cup \chi(N(v))|}{k-\chi(N(v))}$. From \cref{c1AndPrtau>p_chi} in \WarmupDecode{} we have:

\begin{align*}
    \Pr[\chi'(v)=c_1]&=\Pr[c_1\notin \chi(N(v))]\times \Pr[\tau\geq p_{\chi}(v)]\\
    &=\frac{(k-|\chi(N(v))\cup A|)}{(k-\Delta)}\times \frac{k-\Delta}{k-|\chi(N(v))|}\\ &=  \frac{(k-|\chi(N(v)) \cup A|)}{(k-|\chi(N(v))|)}.
\end{align*}

To prove part $(d)$ note that for \WarmupGen{} the operations with non-trivial running time are:
\begin{itemize}
    \item Pick $\tau_\out$ uniformly from $[0,1]$ to compare with $p_\chi$ which is a fraction whose denominator may be represented with at most $O(\log k)$ bits.
    \item Pick $\sigma$ uar from $\mathcal{S}_{A}$.
    \item Pick $c_1$ uar from $[k]\setminus A$.
    \item Updating $L'_\out(v)$ and $M_\out$.
\end{itemize}
With access to fair coins the first, second and the third operations require expected time $O(\Delta\log k)$. The fourth operation requires expected  time $O(\Delta\log k+\log n)$. Hence the expected running time of \WarmupGen{} is $O(\Delta\log k+\log n)$. 
For \WarmupDecode{}, notice that the \emph{If} clause in \cref{loop:compressdecodestart}-\cref{loop:compressdecodeend} takes time $O(\Delta(\log k+\log n))$. The \emph{Else} clause finds the first color in $M[1:\Delta]$ which is not in $\chi(N(v))$: we implement this by first sorting the colors in $\chi(N(v))$ and then doing a binary search in the sorted list, sequentially for every color in $M[1:\Delta]$. Thus we conclude that the running time of \WarmupDecode{} is $O(\Delta(\log\Delta\log k + \log n))$.
\end{proof}

\subsubsection{Contract updates}
\label{subsec:BoundedListUpdates}

\newcommand{\BLSGen}{\textup{\textsc{Contract.gen}}}
\newcommand{\BLSDecode}{\textup{\textsc{Contract.decode}}}
\newcommand{\BLSname}{\text{Contract}}

In this section we describe the tuples $\alpha$ of the type {contract} which reduce the list size at some vertex to $\leq 2$, and with  significant probability, make the list size $1$.
This type of updates will be applied both in the collapsing and the coalescing phase.
As in \cref{subsec:WarmUpUpdates}, we will present two procedures: $\BLSGen{}$ and $\BLSDecode{}$. 
The procedure $\BLSGen{}$ takes as input a tuple  $\alpha_{\init}=(v_{\init},\tau_\init,L_\init,L'_\init, \\ M_\init)$ and a vertex $v$ with the promise that $|S_{L'_{\init}}(v)|<k-\Delta$, and returns a random tuple $\alpha_{\out}$ of type contract. The update operation corresponding to such a tuple is obtained using procedure $\BLSDecode{}$, which takes a tuple $\alpha_\out$ (produced by $\BLSGen{}$) and a coloring $\chi\sim L'_{\init}$, and produces another coloring, say $\chi'$. Thus, the update operation $U_{\alpha_\out}$ is the map $\chi \mapsto \BLSDecode[\alpha_\out, \chi]$. The following lemma describes the relationship between the two procedures, and their important properties.

\begin{lemma}
\label{lm:BoundedListSampler}

Let $\alpha_{\init}=(v_{\init},\tau_\init,L_\init,L'_\init, M_\init)$ be an arbitrary 5-tuple and $v \in V$. Suppose $|S_{L_\init}(v)|<k-\Delta$ and that $\alpha_{\out}=(v_{\out},\tau_{\out},L_{\out}, L'_{\out},M_{\out})$ is a tuple produced by $\BLSGen[\alpha_{\init},v]$. Then
\begin{enumerate}[label=(\alph*)]
    \item Let $L=L'_{\init}$. Then $L_{\out}=L$, $L'_{\out}(u)=L(u)$ for all $u \neq v$, and with probability $p_{L}= 1-\frac{|S_{L}(v)|-|Q_{L}(v)|}{k-\Delta}$ we have $|L'_{\out}(v)|=1$. With the remaining probability we have $|L_{\out}'(v)|=2$.
    \item   For all $\chi\sim L_{\out}$, we have $\chi'\coloneqq \BLSDecode[\alpha_{\out},\chi] \sim L'_{\out}$ (with probability $1$).
    \item  For all $\chi\sim L_{\out}$, the coloring $\chi'$ has the same distribution as $\gd(\chi,v)$.
    \item Except for copying of the list $L'_\init$ the expected running time of \BLSGen{} is $O(\Delta(\log k + \log n))$. The time needed to update a $k$-coloring $\chi$ using \BLSDecode{} is $O(\Delta(\log n +\log k))$.
\end{enumerate}
\end{lemma}

\begin{algorithm}
\caption{\textsc{\BLSname}: generation and decoding}\label{alg:boundedlistsampler}
\setstretch{1.17}
\SetKwInOut{KwIn}{Input}
\SetKwInOut{KwOut}{Output}
\SetKwFunction{gen}{gen}
\SetKwFunction{decode}{decode}
\SetKwProg{myproc}{Function}{}{}
\myproc{\gen{}\rm{:}}{
 \KwIn{$\alpha_{\init}=(v_{\init},\tau_{\init},L_{\init},L'_{\init}, M_{\init})$, $v\in V$ with $|S_{L_\init}(v)|<k-\Delta$}
\KwOut{$\alpha_{\out}=(v_{\out},\tau_{\out},L_{\out},L'_{\out},M_{\out})$}
$L\gets L'_{\init}$; ~$\tau_{\out} \xleftarrow{R}[0,1]$\;\label{L-init}

$c_1\xleftarrow{R}[k]\setminus S_L(v)$;~ $c_2\xleftarrow{R} S_L(v)\setminus Q_L(v)$ \; $p_{L}\gets 1-\big(|S_L(v)|-|Q_L(v)|\big)~/~\big(k-\Delta\big)$ \label{prob:p_L} \tcp*{ $\abs{S_L(v)}<k-\Delta$ ensures $p_L\in[0,1]$}
\eIf{$\tau\leq p_{L}$\label{loop:BLSstart}}
   {$L'_{\out}(v)\gets \brak{c_1}$;~
   $M_{\out}\gets (c_1)$;
   }
   {
   $L'_{\out}(v)\gets \brak{c_1,c_2}$;~
   $M_{\out}\gets (c_1,c_2)$;
   }\label{loop:BLSend}
\KwRet{$\alpha_{\out}=(v,\tau_{\out},L'_{{\init}},L'_{\out},M_{\out})$}
}
\myproc{\decode{}\rm{:}}{
\KwIn{$\alpha=(v,\tau,L,L',M)$, $\chi\sim L'$ with $\abs{S_L(v)}<k-\Delta$, \\ $M[1]\notin S_L(v)$ and $M[2]\in S_L(v)$ or $M[2]=\emptyset$ }
\KwOut{$\chi' \sim L'$}
$\chi' \gets \chi$\;
$p_{\chi}\gets 1-\frac{\abs{S_L(v)}-\abs{Q_L(v)}}{k-\abs{\chi(N(v))}}$  \tcp*{$\abs{S_L(v)}<k-\Delta\leq k-\abs{\chi(N(v))}$ ensures $p_\chi \in[0,1]$}\label{prob:p_chi}
\label{loop:contractloopstart}
\uIf{$\tau \leq p_{\chi}$ or $M[2]\in \chi(N(v))$ \label{loop:DecodeBLSend}}
  {$\chi'_{\out}(v)\gets M[1]$\;} 

\Else{
   $\chi'_{\out}(v)\gets M[2]$\;
    } \label{loop:contractloopend}
\KwRet{$\chi'$} }


     
\end{algorithm}

We now describe the ideas behind \BLSGen{} and \BLSDecode{}. Consider   $\alpha_{\init}=(v_{\init},\tau_\init,L_\init,L'_\init, M_\init)$ and let $L=L'_{\init}$. Also consider a vertex $v\in V$ and a coloring $\chi\sim L$. We wish to produce $\alpha_{\out}=(v_{\out},\tau_{\out},L_{\out}, L'_{\out},M_{\out})$ with $L'_{\out}(v)$ of size at most $2$ (with some probability of it being of size $1$) and a coloring $\chi'$ such that $\chi'\sim L'_{\out}$, and  $\chi'$ should be distributed according to \gd{$(\chi,v)$}. 
For now let us focus on producing $L'_{\out}(v)$ of size $2$.
Hence, without knowing $\chi(N(v))\coloneqq \chi(N(v))$ we need to produce two colors $L'_{\out}(v)=\set{c_1,c_2}$ such that by choosing one of them (based on the
coloring $\chi$), we may ensure that $\chi'(v)$ is distributed uniformly over $[k]\setminus \chi(N(v))$. Notice that since $\chi\sim L$ we have $Q_L(v)\subseteq \chi(N(v))\subseteq S_L(v)$.
An initial attempt is to sample a color $c_1\notin S_{L}(v)$ uar and insist that the update operation set $\chi'(v)=c_1$ (no matter what $\chi$ is). While this is a valid choice of color at $v$ it is not necessarily distributed uniformly over $[k]\setminus \chi(N(v))$, because such an update places no mass on colors in $S_L(v)\setminus \chi(N(v))$. To remedy this situation we sample another color $c_2$ from $S_L(v)\setminus Q_L(v)$ uar and allow the update to choose between $c_1$ and $c_2$ depending on $\chi$. In particular, we prescribe the update at $v$ as follows. Let $\tau_\out$ be chosen from $[0,1]$ uar and let $p_\chi$ be a threshold in $[0,1]$: if $\tau_\out \leq p_\chi$ or $c_2\in \chi(N(v))$ then $\chi'(v)=c_1$; else $\chi'(v)=c_2$. 
Now, it is a matter of calculation to arrange for an appropriate value of $p_\chi$ such that $\chi'(v)$ is uniform over $[k]\setminus \chi(N(v))$. A direct calculation (see proof of \cref{lm:BoundedListSampler}) shows that $p_\chi=1-\frac{|S_L(v)|-|Q_L(v)|}{k-|\chi(N(v))|}$. 
To ensure that with significant probability $L'_{\out}$ has size $1$, we provide a threshold $p_L$ such that always $p_L\leq p_\chi$. Thus, whenever $\tau_\out\leq p_L$ we set $L'_{\out}=\set{c_1}$.
We let $p_L\coloneqq 1-\frac{|S_L(v)|-|Q_L(v)|}{k-|\Delta|}$. 
The assumption $|S_L(v)|<k-\Delta$ implies that $p_\chi \geq p_L >0$. Note that the threshold $p_{\chi}$ is computed after the actual coloring $\chi$ is available for update. \cref{alg:boundedlistsampler} (which may be skipped) is the code implementing the above ideas along with a proof of \cref{lm:BoundedListSampler}.

\begin{proof}[Proof of \cref{lm:BoundedListSampler}]
Let $L=L'_{\init}$. Note that $Q_L(v)\subseteq \\  \chi(N(v))  \subseteq S_L(v)$ as $\chi\sim L$.
Also, $p_L$ at \cref{prob:p_L} is at most $p_\chi$ at \cref{prob:p_chi} in \cref{alg:boundedlistsampler}.
To prove part $(a)$, observe that line \ref{L-init} directly implies $L_{\out}=L$. It is also evident that throughout the execution of the algorithm, $L'_{\out}(u)$ is never updated, for all $u\neq v$, after the execution of line \ref{L-init}. This proves that $L'_{\out}(u)=L(u)$ for all $u\neq v$. Finally, we note that $\tau_{\out}$ is distributed uniformly at random in $[0,1]$ and  $L'_{\out}(v)$ is singleton iff $\tau_{\out}\leq p_{L}$.
Thus $L'_{\out}(v)$ is a singleton with probability exactly $p_L$. 
The claim follows by noting the value of $p_L$ in line \ref{prob:p_L}. 

To prove part $(b)$, note from \BLSGen{} that the set $M_{\out}$ (disregarding the ordering) is actually the same as the set $L'_{\out}(v)$ (Line \ref{loop:BLSstart}-\ref{loop:BLSend}). Also, we see from \BLSDecode{} that $\chi'_{\out}(v)$ is always contained within $M_{\out}$. For all $w\neq v$,  \BLSDecode{} sets $\chi'_{\out}(w)$ to $\chi(w)$. By the hypothesis of the Lemma, $\forall w\in V$, $\chi(w)\in L_{\out}(w)$. Finally, since \BLSGen{} sets $L'_{\out}(w)$ to $L(w)$ for all $w\neq v$, we conclude that the part (b) of the claim is true.

To prove part $(c)$, we first note that the random process \gd{($\chi,v$)} recolors the vertex $v$ with a color chosen uar from the set $[k]\setminus \chi(N(v))$, while retaining the color of every other vertex.
Observe that:
\begin{itemize}
\item[(i)] The distribution of $\chi'(v)$ (induced by Algorithm \ref{alg:boundedlistsampler}), is a convex combination of two uniform distributions : one supported on the set $[k]\setminus S_{L}(v)$ (when $\chi'(v)=c_1$) and the other supported on the set $S_{L}(v)\setminus \chi(N(v))$ (when $\chi'=c_2$). 

\item[(ii)]
$\chi'(v)=c_1$ iff either $\tau\leq p_\chi$ or $c_2\in \chi(N(v))$. Hence,
\[
\Pr[\chi'(v)=c_1]=p_{\chi}+(1-p_{\chi})\cdot\Big( \frac{\abs{\chi(N(v))}-\abs{Q_{L}(v)}}{\abs{S_{L}(v)}-\abs{Q_{L(v)}}}\Big).
\]
\end{itemize}

Observation (i) implies that if $\Pr[\chi'(v)=c_1]$ turns out to be of the form $\frac{k-\abs{S_{L}(v)}}{k-\abs{\chi(N(v))}}$, it would imply that the distribution of $\chi'(v)$ is indeed uniform on the set $[k]\setminus \chi(N(v))$. Referring to \BLSDecode{} and solving for $\Pr[\chi'(v)=c_1]$ by substituting $p_{\chi}$ in observation [b], we verify that this is indeed true. This proves part $(c)$. 


To prove part $(d)$, notice that in \BLSGen{}, the only operations with non-trivial running time consist of :
\begin{itemize}
    \item Pick $\tau_\out$ uniformly from $[0,1]$ to compare with $p_L$ which is a fraction whose denominator may be represented with at most $O(\log k)$ bits
    \item Pick $c_1$ uniformly from $[k]\setminus S_L(v)$
    \item Pick $c_2$ uniformly from $S_L(v)\setminus Q_L(v)$
    \item Updating the list $L'_\out(v)$ and $M_\out$.
\end{itemize}

With access to fair coins the expected running time of \BLSGen{} is $O(\Delta(\log k+\log n))$.
For the running time of \BLSDecode{}, recall that we only update the color at the vertex $v$: the checking and update together take time $O(\Delta(\log k+\log n))$ which concludes the proof.
\end{proof}

\subsection{Collapsing phase}\label{collapsephase}

The collapsing phase will run for $T-T'$ steps from time $t=-T$ to $t=-T'$. The goal of this phase is to bring the list size at every vertex to at most $2$. During the collapsing phase we will generate a sequence $(\alpha_{-T},\ldots,\alpha_{-T'-1})$. Once the corresponding updates are applied, the list sizes of all the vertices will be brought down to at most two. As mentioned in the introduction, this reduction in list size will be achieved by updates of type {contract}; each such update will be preceded by a sequence of updates that spruce up the neighbourhood of the vertex whose list we wish to contract. To spruce up the neighborhood (recall that this happens when the union of the lists of the neighbors has size less than $k-\Delta$ ), we will repeatedly use compression; recall that the \textsc{compress} primitive accepts a set of colors $A$ of size $\Delta$ and ensures that, after the update is applied, the list of the updated vertex has at most one color outside $A$. While doing so, we must ensure that the lists of vertices that have already been collapsed are not disturbed. Let $V=\brak{v_1,\ldots,v_n}$, $N_{>}(v_i)\coloneqq \set{v_j\in N(v)\mid j>i}$ and $N_{<}(v_i)\coloneqq \set{v_j\in N(v)\mid j<i}$. To spruce up the  neighborhood of $v_i$, we will compress the lists of vertices in $N_{>}(v_i)$ and not in $N_{<}(v_i)$. Yet we need to ensure that the entire neighborhood, i.e., $N_{<}(v_i)\cup N_{>}(v_i)$ is spruced up; so the set for sprucing up the neighborhood of $v_i$, $A$ will be chosen such that it includes at least one element from the list of each $w \in N_{<}(v_i)$. The following code implements this.
\begin{algorithm}
\setstretch{1.36}
\SetKwInOut{KwIn}{Input}
\SetKwInOut{KwOut}{Output}
\KwIn{$\alpha_{\init}=(v_{\init},\tau_{\init},L_{\init},L'_{\init},M_{\init})$, $i\in [n]$}
\textbf{Promise:} for all $j<i$: $|L'_{\init}(v_j)|\leq 2$\\
\KwOut{$\alpha[-1,-|N_{>}(v_i)|]$}
$t\gets -|N_{>}(v_i)|$; $L\gets L'_{\init}$ \;
Pick a $\Delta$-element subset $A$ of $[k]$ that intersects every set in 
$\{L(w): w\in N_{<}(v_i) \}$\;  \label{line:maximallyintsubset}
\For{$w\in N_{>}(v_i)$\label{loopwarmuphelperstart}}
    {
        $\alpha_{t}\gets \WarmupGen{[\alpha_{t-1},w,A]}$\;
        $t\gets t+1$\;
    }\label{loopwarmuphelperend}

\KwRet{$\alpha[-1,-|N_{>}(v_i)|]$}
\caption{\warmuphelper{}} \label{alg:warmuphelper}
\end{algorithm}

\begin{lemma}
\label{lm:warmuphelper}
Let $\alpha[-1,-|N_{>}(v_i)|]$ be the output of the algorithm \warmuphelper{$[\alpha_{\init},i]$} and let $\alpha[{-1}]=(v,\tau,L',L'',M)$. Then, $(a)$ for $w\not\in N_{>}(v_i)$: $L''(w)= L'(w)$; $(b)$
$\bigcup_{w \in N(v_i)} |L(w)|\leq 2\Delta$.
\end{lemma}
\begin{proof}
Part $(a)$ is true because the list of no vertex outside $\set{v_i}\cup N_{>}(v_i)$ is perturbed by the algorithm.
Part $(b)$ follows as the at most $\Delta$ neighbors of $v_i$ can each contribute to the union at most one color outside the set $A$ .
\end{proof}

\newcommand{\last}{\textup{\textsc{last}}}
By successively sprucing up the neighborhood and contracting the lists of all vertices in $V$, we complete the collapsing phase. The following code implements this formally; here we adopt the notation, that if $\alpha$ is a sequence of update tuples, then $\alpha[\last]$ is the tuple in this list corresponding to the latest update.
 
\begin{algorithm}
\setstretch{1.36}
\SetKwInOut{KwIn}{Input}
\SetKwInOut{KwOut}{Output}
\KwOut{$\alpha[-1,-(T-T')]$}
$\alpha_{\last} \leftarrow (v_1,0,[k]^V,[k]^V,())$ \;
$\alpha \leftarrow \text{empty}$ \;
\For{$i=1,2,\ldots,n$\label{looplastwarmloopstart}}{
$\alpha \gets    \warmuphelper{[\alpha_\last,i]} \circ \alpha$\;
$\alpha_\last \gets \BLSGen{[\alpha[\last],v_i]}$\;
$\alpha \gets   \alpha_\last \circ \alpha$\;
}\label{looplastwarmloopend} 
\KwRet{$\alpha$}
\caption{\textsc{collapse}} \label{alg:warmup}
\end{algorithm}
\begin{lemma}\label{lm:warmup}
The collapsing phase lasts for $T-T'=|E(G)|+n$ steps.
Let $\alpha[-1,-(T-T')]$ be the output of \textup{\warmup[\empty]} and let $\alpha_{-1}=(v,\tau,L',L'',M)$.
Then, for all $w\in V$ we have $|L''(w)|\leq 2$.
\end{lemma}
\begin{proof}
 When $v_i$ is chosen for update, only $v_i$ and those of its neighbors which succeed it in the ordering are updated. Hence, the edges joining $v$ to these neighbors are counted only this one time with the update. Thus, we set $T- T'=|E(G)|+n$, where $E(G)$ is the edge set of the graph $G$. The remaining part of the claim is obvious based on previous lemmas.
\end{proof}

\subsection{Coalescence phase}
\label{subsec:2.3}
The collapsing phase produces a random sequence of updates, say $\alpha$, at the end of which the lists of all vertices have size at most two. We now propose to follow this up by a another sequence $(\beta_{-1},\ldots,\beta_{-T'})$ and ensure that $|L_0(v)|=1$ for all $v\in V$, with probability at least $1/2$. As stated in the introduction, this is achieved by applying contracting updates $T'$ times at vertices chosen uniformly at randomly. More precisely, let $w[-1,-T']=(w_{-T'}, w_{-T'+1}, \ldots, w_{-1})$ be chosen uniformly from $V^{T'}$: the random sequence $\beta[-1,-T']$ is obtained using the random process $\beta[-T']\leftarrow \BLSGen{[\alpha[\last],w_{-T'}]}$, and $\beta[-i+1]\leftarrow \BLSGen{[\beta[-i], w_{i+1}]}$, for $i=-T',-T'+1,\ldots, -2$. Note that after the collapsing phase all the neighborhoods are spruced up (since each list is of size at most $2$), and thus further application of the contract updates leaves all the neighborhoods spruced up, which is the case for the entirety of the coalescence phase.

\begin{algorithm}
\setstretch{1.36}
\SetKwInOut{KwIn}{Input}
\SetKwInOut{KwOut}{Output}
\KwIn{$\alpha_{in}=(v_{\init},\tau_{\init},L_{\init},L'_{\init},M_{\init})$}
\textbf{Promise:} for all $v\in V$: $|L'_{\init}(v)|\leq 2$\\
\KwOut{$\beta[-1,-T']$}
$\alpha_{-T'-1}\gets \alpha_{\init}$\;
\For{$t=-T',\ldots,-1$}{
$v \xleftarrow{R} V$\;
$\beta_t\gets \BLSGen{[\alpha_{t-1},v]}$ \label{line:CheckForBLSInCoals} 
}
\KwRet{$\beta[-1,-T']$}
\caption{\textsc{Coalescence}} \label{alg:coalescence}
\end{algorithm}
Recall that after the collapsing updates, the list sizes have a significant probability reducing to $1$ from $2$; this is progress.
However, it can also be the case that when an update is performed at a vertex with list size $1$, its list size become $2$. \\ \cref{lm:BoundedListSampler} shows that if the vertex has many neighbors with singleton lists, then it has a greater chance of acquiring a singleton list; in particular, if all its neighbors have list size $1$, then it definitely acquires a singleton list after the update. To track our progress, we define $W_t\coloneqq\brak{v~|~|L_t(v)|=1}$ (earlier we had defined $W_t$ to be the number of vertices of list size $1$). Then, $|W_t|$ performs a random walk on $[0,n]$. \cref{lm:coalescence} establishes that this walk has a drift towards $n$, and that this walk reaches the absorbing state $n$ with probability at least $1/2$.

\begin{lemma}\label{lm:coalescence}
Assume $k > 3\Delta$ and let $T'=2\frac{k-\Delta}{k-3\Delta}n\ln n$. Suppose the last update of the collapse phase has the form $\alpha[\last]=(v_n,\tau,L,L',M)$ such that $|L'(v)|\leq 2$, for all $v \in V$. Let $\beta[-1,-T']$ be the random sequence of updates generated by the above process for the coalescence phase, starting from $\alpha[\last]$. Suppose $\beta[-1]$ has the form $(.,.,.,L_0,.)$.
Then, with probability at least $1/2$, we have for all $v\in V:$ $|L_0(v)|=1$.
\end{lemma}

To prove this we will require the following claim.

\begin{claim}[{\cite[Theorem~$4$]{Huber98}}]
\label{claim:mart} 
Suppose that $X_t$ is a random walk on $\brak{0,1,\ldots,n}$ where $0$ is a reflecting state and $n$ is an absorbing state. Further $|X_{t+1}-X_{t}|\leq 1$, and $\E[X_{t+1}-X_{t}~|~X_t=i]\geq \kappa_i>0$ for all $X_t<n$.
Let $e_i$ is the expected number of times the walk hits the state $i$. 
Then $$\sum\limits_{i=0}^{n}e_i\leq \sum\limits_{i=0}^{n}\frac{1}{\kappa_i}.$$
\end{claim}

\begin{proof}[Proof of \cref{lm:coalescence}]


For $t=-T',\ldots,-1,0$, let $W_t\coloneqq\brak{v:|L_t(v)|=1}$, let $X_t=|W_t|$ and $\delta_t=X_{t+1}-X_t$. Note that $X_t$ is a random variable based on the random choice of $\beta$. We will use \cref{lm:BoundedListSampler} to establish \begin{equation}
\E[X_{t+1}-X_t\mid X_t] \geq \frac{n-X_t}{n}\left(1-\frac{2\Delta}{k-\Delta}\right).
\label{claim:drift}
\end{equation}
Note that the drift is positive if $2\Delta < k-\Delta$, that is, $k > 3\Delta$. Then, our lemma follows immediately from \cref{claim:drift}, \cref{claim:mart} and Markov's inequality.

It remains to establish \cref{claim:drift}. Let $L$ be
the lists at time $t$. We have the following.
\begin{description}
\item[$X_{t+1}-X_t=1$] iff $w_t$ (the vertex updated in step $t$) has list size $2$ and then its list size becomes $1$ after the update. By \cref{lm:BoundedListSampler}, the last event happens with probability $(1-(|S_L(w)|-|Q_L(w)|)/(k-\Delta))$.

\item[$X_{t+1}-X_t=-1$] iff $w_t$ has list size $1$ and then its list size becomes $2$ after the update. By \cref{lm:BoundedListSampler}, this happens with probability $(|S_L(w)|-|Q_L(w)|)/(k-\Delta)$.
\end{description}
Note that $|S_L(v)|-|Q_L(v)| \leq 2|N(v) \cap \overline{W}_t|$, so
\begin{equation}
\sum_v  |S_L(v)|-|Q_L(v)| \leq 2|\overline{W}_t|\Delta.
\end{equation}
Thus,
\begin{align*}
&\E[X_{t+1}-X_{t}\mid W_t] \\ &=\frac{1}{n}\left[
        \sum_{v \not \in W_t}\left(1-\frac{|S_L(v)|-|Q_L(w)|}{k-\Delta}\right) 
            -\sum_{v \in W_t} \frac{|S_L(v)|-|Q_L(w)|}{k-\Delta}\right]\\
            & =\frac{1}{n}\left[|\overline{W}_t|- \sum_{v\in V} \frac{|S_L(v)|-|Q_L(w)|}{k-\Delta}\right]\\
            & \geq \frac{1}{n}\left[|\overline{W}_t|-  \frac{2|\overline{W}_t|\Delta}{k-\Delta}\right]\\
            & = \frac{n-X_t}{n} \left[1-\frac{2\Delta}{k-\Delta}\right].
\end{align*} 
The claim follows from this.

\end{proof}



\subsection{Proof of \cref{lm:main} and running time analysis of \cref{alg:1}}
\label{subsec:2.5}

Let $\alpha[-T'-1,-T]$ be the random sequence of update tuples of length $T-T'=|E(G)|+n$ produced in the collapse phase; let $\beta[-1,-T']$ be the random sequence of update tuples of length $T'$ produced in the coalescence phase. Our final sequence of tuples
will be $\alpha[-1,T]:=\beta[-1,-T'] \circ \alpha[-T'-1,-T]$, where $T:=|E(G)|+n+T'$.  Let $U[-1,-T]$ be the sequence of updates corresponding to $\alpha[-1,-T]$, obtained by applying the appropriate \emph{decode} procedure to each tuple in $\alpha[-1,-T]$. 
%
The predicate $\Phi$ outputs $\textsc{true}$ iff  $|L_0(v)|=1$.
Let us now justify the various parts of \cref{lm:main} in order.
\begin{enumerate}[label=(\alph*)]
    \item That a sample from $\mathcal{D}$  can be computed efficiently is clear from the fact all sub-routines used in the various algorithms are efficient. That each update instruction $U_i$ can also be computed efficiently is clear from \cref{alg:WarmUpSampler} and \cref{alg:boundedlistsampler}. Also, $\Phi$ is clearly efficiently computable justifying part $(a)$.
    \item Notice that the vertex $v_t$ (either random or fixed) we choose to update in $U_t$ is independent of the evolution up till time $t-1$. Further, at time $t$ whichever sub-routine is used, faithfully follows \gd{} at the vertex $v_t$. Hence, $U(-1,-T)$ takes uniform distributions to uniform distribution justifying part $(b)$.
    \item As mentioned above the predicate $\Phi$ outputs $\textsc{true}$ iff  $|L_0(v)|=1$. We start with $L_T=[k]^V$ and maintain that if at time $t$ there is a coloring $\chi\sim L_{t}$ then $U_t(\chi)\sim L_{t+1}$. This justifies part $(c)$.
    \item Part $(d)$ follows immediately from \cref{lm:coalescence}.
\end{enumerate}

Finally, we justify \cref{thm:main} by analyzing the expected running time of \cref{alg:1}.
Let $i$ be the first index where $\Phi(U[-iT-1, -(i+1)T])=\textup{\textsc{true}}$) and let $\chi = L'_{-iT-1}$ be the unique coloring in the image of $U(-iT-1,-(i+1)T)$. 
Notice that a particular block of updates $U[-jT-1,-(j+1)T]$ (where $j<i$) is processed twice by \cref{alg:1}: once during generation of $U[-jT-1,-(j+1)T]$ and once while computing $U(-1,-iT)(\chi)$.
For applying the function $U(-jT-1,-j(T+1))$, we need to invoke both $\WarmupDecode{}$ and $\BLSDecode{}$ which take as input a tuple $\alpha=(v,\tau,L,L',\\ M)$ and a coloring $\chi$. Notice that both these procedures actually never require the lists $L$ and $L'$. Hence, during the generation of $U[-1,-(i+1)T]$ which corresponds to generating the sequence $\alpha[-1,-(i+1)T]$, we implement the changes performed to the lists by $\WarmupGen{}$ and $\BLSGen{}$ in-place without creating new lists. Thus, during the execution of $\WarmupGen{}$ and $\BLSGen{}$ we skip the step of copying the lists.

To calculate the expected time needed to generate $U[-jT-1,-(j+1)T]$ we analyze the expected time needed to generate the updates corresponding to the two phases.
The collapse phase involves generating $|E(G)|+n$ updates during which we call \WarmupGen{} $|E(G)|$ times and \BLSGen{} $n$ times.  
Recall from part $(d)$s of \cref{lm:NewWarmUpSampler} and \cref{lm:BoundedListSampler} that the expected running times for both \WarmupGen{} and \BLSGen{} are $O(\Delta(\log k + \log n))$. Hence, the expected time needed for the collapse phase is $O(n\Delta^2(\log k+\log n))$.
The coalescence phase involves calling \BLSGen{} $2\frac{k-\Delta}{k-3\Delta}n\ln n$ times. Hence, the expected time needed for the coalescence phase is $O((n\log n)\Delta^2(\log k+\log n))$ and the overall expected time for generating $U[-jT-1,-(j+1)T]$ is $O((n\log n)\Delta^2(\log k+\log n))$.

To calculate the time needed to decode $U(-jT-1,-(j+1)T)$ observe that we need to call \WarmupDecode{} $|E(G)|$ times and \BLSDecode{} $n+2\frac{k-\Delta}{k-3\Delta}n\ln n$ times. 
Recall from part $(d)$s of \cref{lm:NewWarmUpSampler} and \cref{lm:BoundedListSampler} that the expected running times for both \WarmupDecode{}{} and \BLSDecode{}{} are $O(\Delta(\log\Delta\log k + \log n))$.
Hence, the running time of decoding $U(-jT-1,-(j+1)T)$ is $O((n\log n)\Delta^2(\log\Delta\log k +\log n))$.

By part $(d)$ of \cref{lm:main} on expectation the value of $j$ is $2$ the overall expected running time of \cref{alg:1} is $O((n\log^2 n)\cdot (\Delta^2\log\Delta\log k))$.

\newcommand{\kay}{k}
\section{Bottleneck for achieving {$\protect\kay>2\Delta$}}

As the coupling proofs for efficient approximate sampling of colorings work all the way to the bound of $k>2\Delta$, it seems natural to ask if we can obtain an efficient perfect sampler which works with $k>2\Delta$.
In this current framework we have two primitives namely, \warmup{} and \sampler{}, which are the workhorses of our algorithm for $k>3\Delta$. Now, suppose we shoot for a better bound and work with $k\leq 3\Delta$ (even $k=2\Delta+1$). In this case we face two hurdles. 

Firstly, after the application of compress updates to spruce up the neighborhood of a vertex $v$ we are able to guarantee that $|S_L(v)|\leq 2\Delta$: however, if $k\leq 3\Delta$ this is not enough to meet the input requirement  for \sampler{}, i.e., $|S_L(v)|<k-\Delta$. 

Secondly, even if we somehow manage to apply \sampler{} we may still produce lists of size $2$. Recall, that the drift analysis of $|W_t|$ (where $W_t$ is the number of vertices with list size $1$) requires an extra margin of $\Delta$, over the product of the maximum degree ($\Delta$) and the bound on the list sizes we can guarantee, in $k$. If $k\leq 3\Delta$ then we do not have this margin.

\section*{Acknowledgements}
We are grateful to Piyush Srivastava for introducing us to this problem and for the numerous detailed discussions which led to this work. We also thank him for pointing us to the relevant references.
Finally, we are indebted Prahladh Harsha and Jaikumar Radhakrishnan with whom we had many discussions that helped us organize our ideas and put them in a presentable form that is the current write-up. 

{
\bibliographystyle{prahladhurl}
\bibliography{main}
}

\end{document}